\documentclass[a4paper,12pt]{article}

\makeatletter
 
  \@addtoreset{equation}{section}
 \makeatother

\usepackage{amsfonts}
\usepackage{amssymb}
\usepackage{mathrsfs}
\usepackage{amsmath,amsthm}
\usepackage{amsmath}

\usepackage{setspace}

\makeatletter
\def\@biblabel#1{}
\makeatother

\usepackage{cite}

\usepackage{graphicx}
\usepackage{color}

\usepackage{here}

\usepackage{indentfirst}

\begin{document}

\newtheorem{theorem}{Theorem}
\newtheorem{lemma}{Lemma}
\newtheorem{prop}{Proposition}
\newtheorem{cor}{Corollary}
\theoremstyle{definition}
\newtheorem{defn}{Definition}
\newtheorem{remark}{Remark}
\newtheorem{step}{Step}

\newcommand{\Cov}{\mathop {\rm Cov}}
\newcommand{\Var}{\mathop {\rm Var}}
\newcommand{\E}{\mathop {\rm E}}
\newcommand{\const }{\mathop {\rm const }}
\everymath {\displaystyle}

\newcommand{\ruby}[2]{
\leavevmode
\setbox0=\hbox{#1}
\setbox1=\hbox{\tiny #2}
\ifdim\wd0>\wd1 \dimen0=\wd0 \else \dimen0=\wd1 \fi
\hbox{
\kanjiskip=0pt plus 2fil
\xkanjiskip=0pt plus 2fil
\vbox{
\hbox to \dimen0{
\small \hfil#2\hfil}
\nointerlineskip
\hbox to \dimen0{\mathstrut\hfil#1\hfil}}}}

\def\qedsymbol{$\blacksquare$}

\renewcommand{\refname }{References}

\title{
An optimal execution problem in the volume-dependent Almgren--Chriss model
}

\author{Takashi Kato
\footnote{Association of Mathematical Finance Laboratory (AMFiL), 
              2--10 Kojimachi, Chiyoda, Tokyo 102-0083, Japan.
E-mail: \texttt{takashi.kato@mathfi-lab.com}}
}

\date{First Version: January 31, 2017\\
This Version: August 24, 2017}

\maketitle
\thispagestyle{empty}

\begin{abstract}
In this study, we introduce an explicit trading-volume process into the Almgren--Chriss model, which is a standard model for optimal execution. We propose a penalization method for deriving a verification theorem for an adaptive optimization problem. We also discuss the optimality of the volume-weighted average-price strategy of a risk-neutral trader. Moreover, we derive a second-order asymptotic expansion of the optimal strategy and verify its accuracy numerically. 
\end{abstract}

\everymath {\displaystyle}

\section{Introduction}
Optimal execution problems have been of considerable interest in the field of mathematical finance during the past two decades. \cite{Bertsimas-Lo} wrote the seminal paper on optimal execution, and the model introduced by \cite{Almgren-Chriss} is known as a standard model in both theory and practice. \cite{Gatheral-Schied} provides a survey of dynamical models that address execution problems.

When studying execution problems, we must consider the market impact (MI), which is the effect that a trader's investment behavior has on security prices. Several studies have proposed optimal execution models with permanent/temporary MI functions (for details of permanent/temporary MIs, see \cite{Almgren-Chriss, Gatheral-Schied}, and \cite{Holthausen-et-al}, for instance).

Market trading volume (turnover), as a representative index of financial market activity, is another important factor in execution problems. This is despite the fact that several classic studies were not much concerned with it. If the trading volume is high, the security is highly liquid and a trader can liquidate shares of the security easily. As an execution strategy that exploits trading volume, the volume-weighted average price (VWAP) strategy is well known and widely used in practice (see \cite{Madhavan}). The VWAP strategy is an execution strategy whose execution speed is proportional to the trading volume of the relevant security. \cite{Frei-Westray, Gueant-Royer}, and \cite{Konishi} considered how to minimize VWAP slippage (i.e., the replication cost of the VWAP strategy) problems as a type of stochastic control problem.

Although the VWAP strategy is a standard execution strategy, it remains unclear why it is effective in terms of the mathematical theory of optimal execution. One rationalization is to consider execution problems on a ``volume-weighted time line.'' \cite{Gatheral-Schied} state in their Remark~22.7 that ``[a strategy with constant execution speed] can be regarded as a VWAP strategy, $\ldots $ the time parameter $t$ does not measure physical time but volume time, which is a standard assumption in the literature on order execution and market impact.'' However, we should not ignore the uncertainty of the market-trading-volume process: we cannot capture how many shares of a security will be traded until a future time horizon.

\cite{Kato_VWAP_Preprint} introduced an explicit market-trading-volume process as a stochastic process and showed that an optimal strategy for problems involving the minimization of expected execution cost is actually the VWAP strategy. This was for an MI function of a general shape and for fluctuations in security price that were given by the Black--Scholes model.

In contrast, \cite{Kato_JSIAM_VWAP} investigated whether the VWAP strategy is optimal in the generalized Almgren--Chriss (AC) model equipped with volume-dependent temporary MI functions. 
As stated in that study, the VWAP strategy is optimal when the trader is risk neutral and our admissible strategies are static (i.e., deterministic) or anticipatory (i.e., depending on future information). However, only limited attention was paid to extending the argument to a standard adaptive optimization problem.

In this study, we propose a simple penalization method to provide a verification-type theorem for the adaptive optimization problem in the generalized AC model. As an example, we give 
a generalized version of the result obtained heuristically by \cite{Kato_JSIAM_VWAP}. This result states that the expected VWAP strategy is optimal in the time-varying Black--Scholes framework.

We also provide a second-order asymptotic expansion formula for the adaptive optimal strategy (say $\hat{\bf x}^\mathrm {adap}$). We compare numerically the performance of $\hat{\bf x}^{\mathrm {adap}}$ with the expected/exact VWAP strategies defined by \cite{Kato_JSIAM_VWAP} when the market trading volume follows the geometric Ornstein--Uhlenbeck (OU) process. Here, the expected (resp., exact) VWAP strategy, say $\hat{\bf x}^{\mathrm {stat}}$ (resp., $\hat{\bf x}^{\mathrm {ant}}$), is a solution to the static (resp., anticipatory) optimal execution problem. Our numerical results are summarized in Figs.~1--4. We find that when the mean reverting speed of the trading volume is low, the fluctuation of $\hat{\bf x}^{\mathrm {adap}}$ is quite similar to that of 
$\hat{\bf x}^{\mathrm {ant}}$ until close to the terminal time period (note that the fluctuation of $\hat{\bf x}^{\mathrm {ant}}$ is completely proportional to that of the market-trading-volume process: see (\ref{def_exact_VWAP}) in Section~\ref{sec_VWAP}). 
This also means that the form of $\hat{\bf x}^{\mathrm {adap}}$ differs from that of $\hat{\bf x}^{\mathrm {stat}}$: $\hat{\bf x}^{\mathrm {stat}}$ is deterministic and does not fluctuate randomly. However, the execution cost corresponding to $\hat{\bf x}^{\mathrm {adap}}$ is not much different from that incurred by $\hat{\bf x}^{\mathrm {stat}}$, which is easy to implement in practice. In contrast, when the mean reverting speed is high, we see not only that the form of $\hat{\bf x}^{\mathrm {adap}}$ is clearly different from that of $\hat{\bf x}^{\mathrm {stat}}$, but that adaptive optimization leads to a lower execution cost than that with static optimization. These results imply that the trader has an incentive when constructing $\hat{\bf x}^{\mathrm {adap}}$ to expend efforts due to improve the execution algorithm when the market-trading-volume process has a strong mean reverting property.

The rest of this paper is organized as follows. In Section~\ref{sec_model}, we introduce our basic model settings. In Section~\ref{sec_VWAP}, we list the definitions of VWAP execution strategies and discuss their optimality. In Section~\ref{sec_main}, we present our main results. We conclude this paper in Section~\ref{sec_conclusion}. Appendix~\ref{sec_proofs} summarizes the proof of Theorem~\ref{th_verification}. In Appendix~\ref{sec_permanent}, we discuss a minor generalization in which both permanent and temporary MI functions depend on the market trading volume.

\section{Model settings}\label{sec_model}

In this section, we introduce our model of an optimal execution problem. Our model is based on the AC model proposed by \cite{Almgren-Chriss} and generalized by \\
\cite{Gatheral-Schied_AC1} and \cite{Schied_AC2}.

We assume that there is a financial market that consists of a risk-free asset (called cash) and a risky asset (called a security). The price of cash is fixed as $1$, whereas the price of the security fluctuates randomly.To describe these price fluctuations, we introduce a probabilistic model. Let $T > 0$, let $(\Omega , \mathcal {F}, (\mathcal {F}_t)_{0\leq t\leq T}, P)$ be a stochastic basis, and let $(S^0_t)_{0\leq t\leq T}$ be a c\`adl\`ag $(\mathcal {F}_t)_{0\leq t\leq T}$-martingale satisfying $\E [\sup _{0\leq t\leq T}|S^0_t|] < \infty $. Here, $S^0_t$ is regarded as the unaffected price of the security at time $t$. When there is no MI, we can trade the security for the price $S^0_t$. For brevity, we assume that $\mathcal {F}_0$ is trivial, that is, all $\mathcal {F}_0$-measurable random variables are almost surely constant.

Next, we introduce a single trader who has $X_0$ shares of the security at the initial time $t=0$. The trader has to liquidate all the shares before a fixed time horizon $T > 0$. The execution strategy ${\bf x} = (x_t)_{0\leq t\leq T}$ is given as a stochastic process. Here, $x_t$ denotes the execution speed at time $t$. The remainder of the shares of the security is given as \begin{eqnarray}\label{def_X}
X_t = X_0 - \int ^t_0x_r\, dr. 
\end{eqnarray}

Under the given execution strategy ${\bf x} = (x_t)_{0\leq t\leq T}$, the security price $S_t$ is defined as follows: 
\begin{eqnarray}\label{def_S}
S_t = S^0_t - \int ^t_0g(x_r)\, dr - \tilde{g}(v_t, x_t),
\end{eqnarray}
where $g$ (resp., $\tilde{g}$) is a permanent (resp., temporary) MI function and $(v_t)_{0\leq t\leq T}$ is a positive $(\mathcal {F}_t)_{0\leq t\leq T}$-adapted process describing the instantaneous market-trading-volume process. In this study, $g$ is assumed to be a linear function, namely $g(x) = \kappa x$ for some $\kappa > 0$. This assumption is standard in several related studies, such as \cite{Almgren-Chriss, Bertsimas-Lo, Cheng-Wang, Gatheral-Schied_AC1, Kato_OU}, and \cite{Schied_AC2}. Note that the linearity of $g$ is assumed for its tractability at first, and that
Almgren et al.~(2005a,b) suggest its validity from an empirical point of view. Moreover, in \cite{Gatheral}, it is justified from economic point of view, that is, the market admits no dynamic arbitrage (in other words, no price manipulation) only when $g$ is linear. In contrast, 
Kato (2014a,b) focuses on a nonlinear version of $g$ and constructs a mathematical model of optimal execution with convex/S-shaped MI functions. The nonlinearity of $g$ is also discussed by \cite{Alfonsi-Fruth-Schied, Gueant}, and \cite{Kato_FMA} from economic and empirical viewpoints.

The temporary MI function $\tilde{g}$ depends on the execution strategy $x_t$ and the trading volume $v_t$. It is natural that the temporary MI should decrease as the trading volume increases, because a large trading volume implies high market liquidity. One of the simplest such settings is where $\tilde {g}(v_t, x_t)$ is proportional to $x_t$ and inversely proportional to $v_t$. Therefore, for the sake of simplicity, we adopt the following form of $\tilde{g}$:
\begin{eqnarray}\label{tilde_g}
\tilde{g}(v, x) = \frac{\tilde{\kappa }x}{v},
\end{eqnarray}
where $\tilde{\kappa } > 0$. 

Next, we define our objective function. For a given ${\bf x}$, an implementation shortfall (IS) cost is defined as 
\begin{eqnarray}\label{def_cost}
\mathcal {C}({\bf x}) = S_0X_0 - \int ^T_0S_tx_t\, dt.
\end{eqnarray}
Substituting (\ref{def_S}) into (\ref{def_cost}) and applying integration by parts, we get 
\begin{eqnarray}\label{calc_IS_cost}
\mathcal {C}({\bf x}) =
\frac{\kappa X_0^2}{2} -
\int ^T_0X_tdS^0_t + \tilde{\kappa }\int ^T_0\frac{x^2_t}{v_t}\, dt,
\end{eqnarray}
where $(X_t)_t$ is defined in (\ref{def_X}). 

We now define the set of admissible strategies: 
\begin{eqnarray}\label{def_adm_str}
\mathcal {A}(X_0) = \left\{ {\bf x} = (x_t)_{0\leq t\leq T} \ ; \ (\mathcal {F}_t)_{0\leq t\leq T}\mbox {-adapted}, x_t\geq 0 \ \mbox { and } \int ^T_0x_t\, dt = X_0\ \mbox {a.s.}  \right\} .  \ \ 
\end{eqnarray}
Note that the final equality on the right-hand side of the above implies the sell-off condition $X_T = 0$, that is, the trader is prohibited from having any remaining shares of the security at the time horizon $T$. We call an element in $\mathcal {A}(X_0)$ an ``adaptive strategy'' to distinguish it from other types of strategy given below.
Also note that $x_t\geq 0$ implies that ${\bf x}$ does not contain any buy orders. 
See Remark \ref {rem_non_negativity} in Section \ref {sec_main} for details on this condition. 

We are ready to define our optimization problem as the problem of minimizing the expected IS cost: 
\begin{eqnarray}\label{adaptive_optimization}
\hat{J}^\mathrm {adap}(X_0) = \inf _{{{\bf x}}\in \mathcal {A}(X_0)}\E [\mathcal {C}({{\bf x}})]. 
\end{eqnarray}
From (\ref{calc_IS_cost}), we can easily see that the above problem is equivalent to \begin{eqnarray}\label{value_function_adapted}
J(X_0) = \inf _{{{\bf x}}\in \mathcal {A}(X_0)}\E \left[\int ^T_0\frac{x^2_t}{v_t}dt \right] .
\end{eqnarray}
Indeed, it holds that $\hat{J}^\mathrm {adap}(X_0) = \kappa X^2_0/2 + \tilde {\kappa }J(X_0)$.

\section{VWAP strategies}\label{sec_VWAP}

In this section, we briefly introduce VWAP execution strategies. Moreover, we review the results in \cite{Kato_JSIAM_VWAP} to verify the optimality of VWAP strategies in some cases. 

We say that ${\bf x} = (x_t)_{0\leq t\leq T}$ is a VWAP strategy if it holds that 
\begin{eqnarray}\label{def_VWAP_execution}
x_t = \gamma v_t, \ \ t\in [0, T] \ \mbox {a.s.}
\end{eqnarray}
for some $\gamma > 0$. Here, $\gamma $ is called a market involvement ratio. Note that if ${{\bf x}}$ is a VWAP strategy, then the trader's execution VWAP $S^{\mathrm {vwap}}_T({\bf x})$ coincides with the market VWAP $S^{\mathrm {VWAP}}_T$: 
\begin{eqnarray*}
S^{\mathrm {vwap}}_T({\bf x}) = 
\frac{\int ^T_0S_tx_t\, dt}{\int ^T_0x_t\, dt}, \ \ 
S^{\mathrm {VWAP}}_T = 
\frac{\int ^T_0S_tv_t\, dt}{\int ^T_0v_t\, dt}
\end{eqnarray*}
(see \cite{Kato_JSIAM_VWAP} for details).

We note that any (adaptive) admissible strategy ${\bf x}\in \mathcal {A}(X_0)$ cannot become a VWAP strategy in the strict sense. Indeed, if ${\bf x}$ satisfies (\ref{def_VWAP_execution}), the sell-off condition $X_T = 0$ immediately implies that $\gamma = X_0/V_T$, which contradicts the assumption that $x_t\in \mathcal {F}_t$, where $V_t$ is a cumulative trading-volume process: 
\begin{eqnarray}\label{def_Vt}
V_t = \int ^t_0v_r\, dr, \ \ t\in [0, T]. 
\end{eqnarray}
Nevertheless, we place importance on the VWAP strategy as a ``benchmark'' of appropriate execution strategies. Indeed, as stated in Theorem~3 of \cite{Kato_JSIAM_VWAP}, the strategy 
\begin{eqnarray}\label{def_exact_VWAP}
\hat{\bf x}^\mathrm {ant} = (\hat{x}^\mathrm {ant}_t)_{0\leq t\leq T}, \ \ \hat{x}^\mathrm {ant}_t = \frac{X_0v_t}{V_T}, \ \ t\in [0, T]
\end{eqnarray}
is a solution to the ``anticipating'' optimization problem 
\begin{eqnarray*}
\hat{J}^\mathrm {ant}(X_0) = \inf _{{\bf x}\in \mathcal {A}^\mathrm {ant}(X_0)}\E [\mathcal {C}({\bf x})], 
\end{eqnarray*}
where $\mathcal {A}^\mathrm {ant}(X_0)$ is the set of $(\hat{\mathcal {F}}_t)_{0\leq t\leq T}$-adapted processes 
satisfying $\int ^T_0x_t\, dt = X_0$. Here, $(\hat{\mathcal {F}}_t)_{0\leq t\leq T}$ is defined by $\hat{\mathcal {F}}_t = \mathcal {G}_T\vee \mathcal {H}_t$, where $(\mathcal {G}_t)_{0\leq t\leq T}$ (resp., $(\mathcal {H}_t)_{0\leq t\leq T}$) is a filtration generated by $(v_t)_{0\leq t\leq T}$ (resp., $(S^0_t)_{0\leq t\leq T}$)\footnote{
Strictly speaking, to show the optimality of ${\bf x}^{\mathrm {ant}}$ for $\hat{J}^{\mathrm {ant}}(X_0)$, we require an additional condition such that, for instance, $(\mathcal {G}_t)_{0\leq t\leq T}$ and $(\mathcal {H}_t)_{0\leq t\leq T}$ are independent. We omit the details here because the main scope of this paper is the adaptive optimization problem.}.

We call the strategy (\ref{def_exact_VWAP}) an ``exact VWAP strategy.'' If we can use full information on the random variable $V_T$ at time $t = 0$, the exact VWAP strategy is optimal in the sense of minimizing the expected IS cost. However, it is impossible to observe $V_T$ until time $t = T$, and so we cannot implement the exact VWAP strategy in practice.

As a substitute for (\ref{def_exact_VWAP}), we define 
\begin{eqnarray}\label{def_static_VWAP}
\hat{\bf x}^\mathrm {stat} = (\hat{x}^\mathrm {stat}_t)_{0\leq t\leq T}, \ \ \hat{x}^\mathrm {stat}_t = \frac{X_0u_t}{U_T}, \ \ t\in [0, T], 
\end{eqnarray}
where 
\begin{eqnarray*}
u_t = \E \left[ v^{-1}_t\right] ^{-1}, \ \ U_T = \int ^T_0u_t\, dt.
\end{eqnarray*}
Here, $u_t$ gives a harmonic mean of the random variable $v_t$. We call the strategy given by (\ref{def_static_VWAP}) an ``expected VWAP strategy.'' This is a static (i.e., deterministic) strategy, thus we can construct it by using information from only the initial time. Theorem~4 in \cite{Kato_JSIAM_VWAP} implies that the expected VWAP strategy is a solution to the static optimization problem 
\begin{eqnarray*}
\hat{J}^\mathrm {stat}(X_0) = \inf _{{\bf x}\in \mathcal {A}^\mathrm {stat}(X_0)}\E [\mathcal {C}({\bf x})], 
\end{eqnarray*}
where $\mathcal {A}^\mathrm {stat}(X_0)$ is a set of ${\bf x}\in \mathcal {A}(X_0)$ such that $x_t$ is non-random. 

Note that these results do not require any explicit model for the volume process $(v_t)_{0\leq t\leq T}$. Also, for the unaffected price process $(S^0_t)_{0\leq t\leq T}$, we assume only the martingale property. Therefore, the optimality of exact/expected VWAP strategies is robust in the framework of the AC model. 

As for the adaptive optimization problem (\ref{adaptive_optimization})--(\ref{value_function_adapted}), the result in \cite{Kato_JSIAM_VWAP} requires the strong assumption that $(v_t)_{0\leq t\leq T}$ is geometric Brownian motion: 
\begin{eqnarray}\label{GBM}
dv_t = \mu dt + \sigma dB_t, \ \ v_0 > 0. 
\end{eqnarray}
Here, $\mu \in \Bbb {R}$ and $\sigma > 0$ are constants, and $(B_t)_{0\leq t\leq T}$ is one-dimensional $(\mathcal {F}_t)_{0\leq t\leq T}$-Brownian motion. Theorem~6 in \cite{Kato_JSIAM_VWAP} implies, without detailed proof, that the expected VWAP strategy is still optimal as a solution to the adaptive optimization problem, that is 
\begin{eqnarray*}
\hat{J}^\mathrm {adap}(X_0) = \hat{J}^\mathrm {stat}(X_0) = \E [\mathcal {C}(\hat{\bf x}^\mathrm {stat})]. 
\end{eqnarray*}
Therefore, we cannot improve the execution cost by extending the class of admissible strategies from $\mathcal {A}^\mathrm {stat}(X_0)$ to $\mathcal {A}^\mathrm {adap}(X_0)$ in this particular case. 

\begin{remark}\label{rem_tilde_g_general}
Note that the above results can be generalized as the case in which $\tilde{g}$ is given as 
\begin{eqnarray}\label{tilde_g_general}
\tilde{g}(v, x) = k(v)x^\alpha , 
\end{eqnarray}
where $\alpha > 0$ is a constant and $k$ is a positive continuous function. In this case, the optimal strategies of the static/anticipating problems are no longer the VWAP strategies. In Remark~5 of \cite{Kato_JSIAM_VWAP}, we call these optimal strategies ``twisted VWAP strategies.'' 
\end{remark}

\begin{remark}\label{rem_TWAP}
When $v_t$ is a constant, the optimal strategy is to sell at a constant rate (i.e., $x_t = X_0/T$; this is 
the exact solution to the optimization problem in the original AC model). This strategy is called a time-weighted average price (TWAP) strategy. In 
Kato (2014a,b, 2016, 2017), we find similar results that show that the TWAP strategy is the optimal strategy for a risk-neutral trader when the permanent MI function is nonlinear. 
\end{remark}

\begin{remark}\label{rem_volume_time}
As mentioned in Remark~\ref{rem_TWAP}, the TWAP strategy is optimal in the original AC model without considering the market trading volume. Here, we give another interpretation of our model by introducing the concept of ``volume time (stochastic clock)'' (see \cite{Ane-Geman, Geman}, and \cite{Veraat-Winkel}) similarly to \cite{Kato_VWAP_Preprint}. 

We use (\ref{def_X}), (\ref{def_cost}), and (\ref{adaptive_optimization}) to define the value function of the problem of minimizing the expected IS cost. However, instead of by (\ref{def_S}), we assume that the security price process is given by 
\begin{eqnarray}\label{def_S_volume_time}
S_t = \tilde {S}^0_{V_t} - \int ^t_0g\left( \frac{x_r}{v_r}\right) dV_r - \hat{g}\left( \frac{x_r}{v_r}\right) , 
\end{eqnarray}
where $(\tilde S^0_{\tilde {t}})_{\tilde{t}\geq 0}$ is a uniformly integrable $(\tilde{\mathcal {F}}_{\tilde{t}\geq 0})_{\tilde{t}\geq 0}$-martingale and the filtration $(\tilde{\mathcal {F}}_{\tilde{t}})_{\tilde{t}\geq 0}$ is given by $\tilde{\mathcal {F}}_{\tilde{t}} = \mathcal {F}_{V^{-1}_{\tilde{t}}}$ (note that $V^{-1}_{\tilde{t}} := \inf \{ t\geq 0\ ; \ V_t\geq \tilde{t} \} \wedge T$ is an $(\mathcal {F}_t)_{0\leq t\leq T}$-stopping time for each fixed $\tilde{t}$). Here, $g$ (resp., $\hat{g}$) is a permanent (resp., temporary) MI function with respect to the instantaneous market-involvement ratio $x_t/v_t$. We recognize the cumulative trading-volume process $V_t$ as the volume time: the quicker $V_t$ increases, the quicker time passes. The unaffected-security-price process is given as a martingale on the volume timeline rather than on the physical timeline. The permanent MI is also accumulated according to the volume time increment $dV_t$ rather than $dt$. 

It is easy to see that the process $(S^0_t)_{0\leq t\leq T}$ defined by $S^0_t = \tilde{S}^0_{V_t}$ is an $(\mathcal {F}_t)_{0\leq t\leq T}$-martingale. Therefore, when $g$ is a linear function, (\ref{def_S_volume_time}) can be rewritten as (\ref{def_S}) by replacing $\tilde{g}(v_t, x_t)$ with $\hat{g}(x_t/v_t)$. This implies that our model 
(\ref{def_X})--(\ref{adaptive_optimization}) can also be regarded as an optimal execution problem in the AC framework with a stochastic clock whenever the permanent MI function is linear
\footnote{Intuitively, the cost due to a permanent MI seems to be small when the trading volume becomes large. Roughly speaking, this intuition is true for convex $g$ but not true for concave $g$. Indeed, we can rewrite the permanent MI term in (\ref{def_S_volume_time}) as $\int ^t_0g(x_r/v_r)v_r\, dr$, and the derivative of the integrand $g(x/v)v$ with respect to $v$ is $g(x/v) - (x/v)g'(x/v)$, which is non-positive (resp., non-negative) when $g$ is convex (resp., concave). Here, we assume the smoothness of $g$ for brevity. Note that if $g$ is linear, the permanent MI term is independent of trading volume. However, even in this case, a large trading volume lessens the fall in price due to the term of the temporary MI.}. 

In Appendix~\ref{sec_permanent}, we study our model from another perspective, that in which the permanent MI function depends explicitly on the trading volume $v_t$. 
\end{remark}

\section{Main results}\label{sec_main}

\subsection{Analytical solution and corresponding verification theorem}\label{sec_verification}

Firstly, we provide a verification theorem that is useful for finding an adaptive optimal execution strategy for the problem (\ref{value_function_adapted}). 

Because the trading-volume process $(v_t)_{0\leq t\leq T}$ is assumed to be always positive, it is useful to describe the dynamics of the log-volume process $Y_t := \log v_t$ rather than those of $v_t$ itself. Therefore, throughout this subsection, we assume that $Y_t$ satisfies the following stochastic differential equation (SDE): 
\begin{eqnarray*}
dY_t = b(t, Y_t)\, dt + \sigma (t, Y_t)dB_t, 
\end{eqnarray*}
where $b, \sigma : [0, T]\times \Bbb {R}\longrightarrow \Bbb {R}$ are Borel-measurable functions. Note that $(v_t)_{0\leq t\leq T}$ satisfies the following SDE: 
\begin{eqnarray*}
dv_t = \hat{b}(t, v_t)\, dt + \hat{\sigma }(t, v_t)dB_t, 
\end{eqnarray*}
where $\hat{b}(t, v) = v(b(t, \log v) + \sigma (t, \log v)^2/2)$ and $\hat{\sigma }(t, v) = v\sigma (t, \log v)$. 

We list the following conditions. 
\begin{itemize}
 \item [\mbox{[A1]}] $b$ and $\sigma $ are bounded and are Lipschitz continuous, that is, there is a positive constant $K$ such that 
\begin{eqnarray*}
&&|b(t, y)| + |\sigma (t, y)| \leq K ,\\
&&|b(t, y) - b(t, y')| + |\sigma (t, y) - \sigma (t, y')| \leq K|y - y'|
\end{eqnarray*}
for each $t\in [0, T]$ and $y, y'\in \Bbb {R}$. 
 \item [\mbox{[A2]}] For each $\lambda > 0$, there exists a function $W^\lambda \in C^{1, 2}([0, T]\times (0, \infty ))$ such that 
\begin{itemize}
 \item [(i)] $W^\lambda $ is a classical solution to the following partial differential equation (PDE): 
\begin{eqnarray}\label{PDE}
\frac{\partial }{\partial t}W^\lambda + 
\hat{b}(t, v)\frac{\partial }{\partial v}W^\lambda + 
\frac{1}{2}\hat{\sigma }(t, v)^2\frac{\partial ^2}{\partial v^2}W^\lambda = v(W^\lambda )^2, \ \ 
W^\lambda (T, v) = \frac{\lambda }{v}; 
\end{eqnarray}
 \item [(ii)] there are positive constants $C_\lambda $ and $m_\lambda $ such that 
\begin{eqnarray}\label{cond_A2ii}
0\leq W^\lambda (t, v) \leq C_\lambda (1 + v^{m_\lambda } + v^{-m_\lambda }). 
\end{eqnarray}
\end{itemize}
 \item [\mbox{[A3]}] There exists $p > 2$ such that 
\begin{eqnarray*}
\E \left[ \int ^T_0\sup _{\lambda > 0}(x^\lambda _t)^pdt \right]  < \infty , 
\end{eqnarray*}
where ${\bf x}^\lambda = (x^\lambda _t)_{0\leq t\leq T}$ is defined as 
\begin{eqnarray}\label{def_penalized_optimizer}
x^\lambda _t = X_0\exp \left( -\int ^t_0v_sW^\lambda (s, v_s)\, ds\right) v_tW^\lambda (t, v_t). 
\end{eqnarray}
\end{itemize}

Then we have the following theorem. 

\begin{theorem}\label{th_verification}
Assume $[A1]$--$[A3]$. Then the limit $x^\infty _t = \lim _{\lambda \rightarrow \infty }x^\lambda _t$ 
exists $dt\otimes dP$-a.e.~and  it holds that ${\bf x}^\infty = (x^\infty _t)_{0\leq t\leq T}\in \mathcal {A}(X_0)$. Moreover, ${\bf x}^\infty $ is an optimizer of $(\ref{value_function_adapted})$, that is, ${\bf x}^\infty $ is the adaptive optimal execution strategy. 
\end{theorem}

The proof of Theorem~\ref{th_verification} is given in Appendix~\ref{sec_proofs}. Note that, as proved in Appendix~\ref{sec_proofs}, for each $\lambda > 0$, ${\bf x}^\lambda = (x^\lambda _t)_{0\leq t\leq T}$ defined by (\ref{def_penalized_optimizer}) is an optimizer of the following stochastic control problem: 
\begin{eqnarray}\label{def_J_lambda0}
J^\lambda (X_0) = \inf _{{\bf x}\in \tilde {\mathcal {A}}(X_0)}\E \left[ 
\int ^T_0\frac{x^2_t}{v_t}dt + \frac{\lambda }{v_T}X^2_T\right], 
\end{eqnarray}
where $\tilde {\mathcal {A}}(X_0)$ is a set of adaptive strategies without the sell-off condition, that is, 
\begin{eqnarray*}
\tilde{\mathcal {A}}(X_0) = \left\{ {\bf x} = (x_t)_{0\leq t\leq T}\ ; \ (\mathcal {F}_t)_{0\leq t\leq T}\mbox {-adapted}, x_t\geq 0 \ \mbox { and } \int ^T_0x_t\, dt \leq  X_0\ \mbox {a.s.}  \right\} . 
\end{eqnarray*}
Moreover, it holds that 
\begin{eqnarray}\label{def_J_lambda}
J^\lambda (X_0) = X_0^2W^\lambda (0, v_0) 
\end{eqnarray}
(see Appendix~\ref{sec_proofs} for details). Note that (\ref{def_J_lambda}) implies that $W^\lambda (0, v_0) = J^\lambda (1)$, hence $W^\lambda $ represents the value function corresponding to the optimization problem (\ref{def_J_lambda0}) when the trader has only one share to sell. 

Obviously, $\mathcal {A}(X_0)$ is a subset of $\tilde{\mathcal {A}}(X_0)$. Hence, it holds that 
\begin{eqnarray}\label{rel_J}
J^\lambda (X_0) \leq J(X_0), \ \ \lambda > 0. 
\end{eqnarray} 
Intuitively, the optimal strategy ${\bf x}^\infty $ for the value function $J(X_0)$ is obtained as a limit of the optimizer ${\bf x}^\lambda $ of $J^\lambda (X_0)$. Therefore, if we find a solution to (\ref{PDE}) for each $\lambda $, we can construct the optimizer of (\ref{value_function_adapted}) explicitly by (\ref{def_penalized_optimizer}) and by letting $\lambda \rightarrow \infty $. Also note that the limit $W^\infty (t, v) \equiv \lim _{\lambda \rightarrow \infty }W^\lambda (t, v)$ exists for each $(t, v)\in [0, T)\times (0, \infty )$ (we should take care that $W^\infty (T, v)$ diverges) and that ${\bf x}^\infty $ satisfies (\ref{def_penalized_optimizer}) if we replace $\lambda $ with $\infty $. Moreover, we have the following ordinary differential equation for the process of the remaining shares $X^\infty _t = X_0 - \int ^t_0x^\infty _s\, ds$: 
\begin{eqnarray}\label{ODE_optimal}
x^\infty _t = -\dot {X}^\infty _t = X^\infty _tv_tW^\infty (t, v_t), \ \ 0\leq t < T. 
\end{eqnarray}
Indeed, a straightforward calculation gives us that 
\begin{eqnarray*}
&&X^\infty _t = X_0 - \int ^t_0x^\infty _s\, ds\\
&&= X_0\left( 1 - \int ^t_0\exp \left( -\int ^s_0v_rW^\infty (r, v_r)\, dr\right) v_sW^\infty (s, v_s)\, ds\right) = 
X_0\exp \left( -\int ^t_0v_tW^\infty (t, v_t)\, dt\right),
\end{eqnarray*}
and hence we can write $x^\infty _t = X^\infty _tv_tW^\infty (t, v_t)$ for each $t\in [0, T)$.

\begin{remark}\label{rem_non_negativity} 
In (\ref{def_adm_str}), we require the non-negativity of $(x_t)_{0\leq t\leq T}$. This implies that we do not consider the possibility of buying the security during the selling program. This setting is natural because our focus is a selling execution problem. Indeed, ${\bf x}^\infty $ in Theorem~\ref{th_verification} is actually non-negative because of the assumption $W^\lambda \geq 0$ in [A2](ii). 

However, there are execution models in which optimal selling execution schedules include purchasing orders (see \cite{Alfonsi-Schied-Slynko} for instance). Furthermore, there is a case in which an optimal execution strategy oscillates between buy and sell orders. Such a problem is related to the concept of ``transaction-triggered price manipulation'' (see Definition~22.2 in \cite{Gatheral-Schied}). Hence, it is meaningful to consider the possibility of negative $x_t$. Moreover, in Appendix~\ref{sec_permanent}, we face a situation in which an optimal selling strategy contains buying orders.

In fact, we can relax the admissibility condition as 
\begin{eqnarray}\label{def_adm_str2}
\mathcal {A}(X_0) = \left\{ {\bf x} = (x_t)_{0\leq t\leq T}\in \hat{\mathcal {A}}(0, X_0)\ ; \ \int ^T_0x_t\, dt = X_0\ \mbox {a.s.}\right\} ,
\end{eqnarray}
where 
\begin{eqnarray}\nonumber 
\hat{\mathcal {A}}(t, X) = \bigg\{ {\bf x} = (x_s)_{t\leq s\leq T} &;& 
(\mathcal {F}_s)_{t\leq s\leq T}\mbox {-adapted}, \mathop {\rm essinf }_{s, \omega }x_s(\omega ) > -\infty  
 \mbox { and } \int ^T_tx_s\, ds \leq X\ \mbox {a.s.}  \bigg\} . \\\label{def_adm_str3}
\end{eqnarray}
Note that for each ${\bf x} = (x_t)_{0\leq t\leq T}$ in (\ref{def_adm_str2}), the process $(X_t)_{0\leq t\leq T}$ defined by (\ref{def_X}) is essentially bounded (see Lemma~\ref{bdd_X} in Appendix~\ref{sec_permanent}). Hence, we exclude strongly oscillating execution strategies from our admissible strategies (\ref{def_adm_str2}).

We adopt (\ref{def_adm_str}) as the class of admissible strategies for brevity, but we stress that our main results are valid when we replace the definition of $\mathcal {A}(X_0)$ with (\ref{def_adm_str2}). 

Similarly, to treat the case in which the optimal strategy ${\bf x}^\infty $ takes a negative value, we can generalize (\ref{cond_A2ii}) as follows: 
\begin{eqnarray}\label{cond_A2ii_gen}
-C'\leq W^\lambda (t, v) \leq C_\lambda (1 + v^{m_\lambda } + v^{-m_\lambda }), 
\end{eqnarray}
where $C'$ is a constant that is independent of $t, v$, and $\lambda $. 

\end{remark}

\subsection{Example: time-dependent Black--Scholes model}\label{sec_eg}

In this subsection, we consider the time-dependent Black--Scholes model, namely the case in which 
\begin{eqnarray}\label{general_BS}
b(t, v) = b_t, \ \ \sigma (t, v) = \sigma _t
\end{eqnarray}
are given as deterministic bounded Borel-measurable functions. 


\begin{theorem}\label{th_eg}
Assume $(\ref{general_BS})$. Then it holds that 
\begin{eqnarray*}
J(X_0) = 
\frac{X^2_0}{v_0}\left( \int ^T_0\exp \left( \int ^t_0(b_s - \sigma ^2_s/2)\, ds \right)\, dt \right)^{-1}.
\end{eqnarray*} 
Moreover, strategy ${\bf x}^\infty = (x^\infty _t)_{0\leq t\leq T}$ defined by 
\begin{eqnarray*}
x^\infty _t = \frac{X_0\exp \left( -\int ^T_t(b_s - \sigma ^2_s/2)\, ds \right) }
{\int ^T_0\exp \left( -\int ^T_s(b_r - \sigma ^2_r/2)\, dr \right)\, ds} 
\end{eqnarray*}
is an optimizer of $(\ref{value_function_adapted})$. 
\end{theorem}

\begin{proof}
We can verify conditions [A1]--[A3] by a straightforward calculation with 
\begin{eqnarray*}
W^\lambda (t, v) &=& 
\frac{1}{v}\left( \int ^T_t\exp \left( \int ^s_t(b_r - \sigma ^2_r/2)\, dr \right)\, ds + 
\frac{1}{\lambda }\exp \left( \int ^T_t(b_s - \sigma ^2_s/2)\, ds \right) \right)^{-1}, \\
J^\lambda (X_0) &=& X^2_0W^\lambda (0, v_0), \\
x^\lambda _t &=& \frac{X_0\lambda \exp \left( -\int ^T_t(b_s - \sigma ^2_s/2)\, ds \right) }
{1 + \lambda \int ^T_0\exp \left( -\int ^T_s(b_r - \sigma ^2_r/2)\, dr \right)\, ds}. 
\end{eqnarray*}
Our assertion is obtained by using Theorem~\ref{th_verification}. 
\end{proof}

Note that ${\bf x}^\infty $ coincides with the expected VWAP execution strategy, and it holds that $\hat{J}^\mathrm {adap}(X_0) = \hat{J}^\mathrm {stat}(X_0)$. 

\begin{remark}\label{rem_nonlinear_tilde_g}
We can generalize the above result to the case in which $\tilde{g}$ is given as (\ref{tilde_g_general}) with $k(v) = \gamma v^{-\beta }$ for some $\beta \geq 0$ and $\gamma > 0$, and the trading-volume process satisfies $v_t = \bar{u}_t\exp \left( \int ^t_0\sigma _sdB_s \right) $ for some continuous positive function $(\bar{u}_t)_{0\leq t\leq T}$ and bounded Borel-measurable function $(\sigma _t)_{0\leq t\leq T}$. In this case, we see that 
\begin{eqnarray*}
x^\infty _t = \frac{X_0\exp \left( \frac{\beta ^2}{2\alpha }\int ^T_t\sigma ^2_s\, ds \right) \bar{u}_t^{\beta /\alpha }}
{\int ^T_0\exp \left( \frac{\beta ^2}{2\alpha }\int ^T_s\sigma ^2_r\, dr \right)\bar{u}_s^{\beta /\alpha }ds}, 
\end{eqnarray*}
which is also equal to the (twisted) expected VWAP strategy, and it holds that $\hat{J}^\mathrm {adap}(X_0) = \hat{J}^\mathrm {stat}(X_0)$. The details are left to the reader. 
\end{remark}

\subsection{Asymptotic expansion for adaptive optimal strategies}

In Section~\ref{sec_verification}, we introduced the verification theorem to facilitate the derivation of an optimizer of (\ref{value_function_adapted}). Moreover, in Section~\ref{sec_eg}, we obtained an analytical solution to the adaptive optimization problem with the generalized Black--Scholes model. However, it is still difficult to find an optimal strategy in the general case. 

If $(v_t)_{0\leq t\leq T}$ is deterministic, the optimal strategy is obviously the expected VWAP strategy $\hat{\bf x}^\mathrm {stat}$. Hence, we consider deriving an asymptotic expansion formula around $\hat{\bf x}^\mathrm {stat}$. Note that the arguments in this subsection are only formal ones; in future work, we intend to seek mathematical justification. 

We consider the following perturbed volume process with a small parameter $\varepsilon > 0$: 
\begin{eqnarray*}
v_t = \bar{u}_t\exp \left( \varepsilon Z^{0, 0}_t \right) , 
\end{eqnarray*}
where $(\bar{u}_t)_{0\leq t\leq T}$ is a deterministic continuous positive function and $(Z^{t, z}_s)_{t\leq s\leq T}$ is a stochastic process that satisfies the following SDE: 
\begin{eqnarray*}
dZ^{t, z}_s = \alpha (s, Z^{t, z}_s)\, ds + \beta (s, Z^{t, z}_s)dB_s, \ \ Z^{t, z}_t = z 
\end{eqnarray*}
for some adequate functions $\alpha (s, z)$ and $\beta (s, z)$. Here, the term $\varepsilon Z^{t, z}_s$ describes a small noise on the trading-volume process. 

Let $W^{\varepsilon , \lambda }(t, z)$ be a classical solution to the following PDE: 
\begin{eqnarray}\label{PDE2}
\frac{\partial }{\partial t}W^{\varepsilon , \lambda } + 
\alpha (t, z)\frac{\partial }{\partial z}W^{\varepsilon , \lambda } + 
\frac{1}{2}\beta (t, z)^2\frac{\partial ^2}{\partial z^2}W^{\varepsilon , \lambda } = \bar{u}_te^{\varepsilon z}(W^{\varepsilon , \lambda })^2, \ \ 
W^{\varepsilon , \lambda } (T, z) = \lambda . 
\end{eqnarray}
Note that $J^{\varepsilon , \lambda }(X_0) = X^2_0W^{\varepsilon , \lambda }(0, 0)$ is given as the following value function: 
\begin{eqnarray*}
J^{\varepsilon , \lambda }(X_0) = \inf _{{\bf x}\in \tilde {\mathcal {A}}(X_0)}\E \left[ 
\int ^T_0\frac{x^2_t}{v_t}dt + \lambda X^2_T\right]. 
\end{eqnarray*}
We see easily that $W^{0, \lambda }(t, z) = W^{0, \lambda }(t) = (\bar{U}_T - \bar{U}_t + 1/\lambda )^{-1}$, where $\bar{U}_t = \int ^t_0\bar{u}_s\, ds$. We consider the formal expansion 
\begin{eqnarray}\label{asy_exp}
W^{\varepsilon , \lambda }(t, z) = W^{0, \lambda }(t, z) + 
\varepsilon I^{1, \lambda }(t, z) + \varepsilon ^2I^{2, \lambda }(t, z) + \cdots 
\end{eqnarray}
for small $\varepsilon > 0$. Substituting (\ref{asy_exp}) for (\ref{PDE2}), we formally obtain 
\begin{eqnarray*}
&&\left( \frac{\partial }{\partial t} + \alpha (t, z)\frac{\partial }{\partial z} + \frac{1}{2}\beta (t, z)^2\frac{\partial ^2}{\partial z^2}\right) 
\left(  W^{0, \lambda } + 
\varepsilon I^{1, \lambda } + \varepsilon ^2I^{2, \lambda } + \cdots \right) \\
&=& 
\bar{u}_t\left( 1 + \varepsilon z + \frac{1}{2}\varepsilon ^2z^2 + \cdots \right) 
\left(  W^{0, \lambda } + 
\varepsilon I^{1, \lambda } + \varepsilon ^2I^{2, \lambda } + \cdots \right)^2 .
\end{eqnarray*}
Expanding both sides and comparing coefficients of $\varepsilon $ and $\varepsilon ^2$, we obtain 
\begin{eqnarray*}
&&\frac{\partial }{\partial t}I^{1, \lambda } + \alpha (t, z)\frac{\partial }{\partial z}I^{1, \lambda } + 
\frac{1}{2}\beta (t, z)\frac{\partial ^2}{\partial z^2}I^{1, \lambda } = 
z\bar{u}_t(W^{0, \lambda }(t))^2 + 2\bar{u}_tW^{0, \lambda }(t)I^{1, \lambda }, \\
&&\frac{\partial }{\partial t}I^{2, \lambda } + \alpha (t, z)\frac{\partial }{\partial z}I^{2, \lambda } + 
\frac{1}{2}\beta (t, z)\frac{\partial ^2}{\partial z^2}I^{2, \lambda }\\&& = 
\bar{u}_t\left\{ \frac{1}{2}z^2(W^{0, \lambda }(t))^2 + 2zW^{0, \lambda }(t)I^{1, \lambda }(t, z) + (I^{1, \lambda }(t, z))^2 \right\} + 2\bar{u}_tW^{0, \lambda }(t)I^{2, \lambda }
\end{eqnarray*}
and $I^{1, \lambda }(T, z) = I^{2, \lambda }(T, z) = 0$. We then apply the Feynman--Kac formula to obtain 
\begin{eqnarray*}
I^{1, \lambda }(t, z) &=& -\E \left[ \int ^T_tZ^{t, z}_s(W^{0, \lambda }(s))^2\bar{u}_s\exp \left( -2\int ^s_tW^{0, \lambda }(r)\bar{u}_r\, dr \right)\, ds  \right] , \\
I^{2, \lambda }(t, z) &=& -\E \Bigg[ \int ^T_t\Big\{ \frac{1}{2}(Z^{t, z}_s)^2(W^{0, \lambda }(s))^2 + 2Z^{t, z}_sW^{0, \lambda }(s)I^{1, \lambda }(s, Z^{t, z}_s) + (I^{1, \lambda }(s, Z^{t, z}_s))^2\Big\} \\
&&\hspace{20mm}\times \bar{u}_s\exp \left( -2\int ^s_tW^{0, \lambda }(r)\bar{u}_r\, dr \right)\, ds  \Bigg] . 
\end{eqnarray*}
Letting $\lambda \rightarrow \infty $, 
we get the following formal expansion formula: 
\begin{eqnarray}\label{asy_exp_W}
W^{\varepsilon }(t, z) = \frac{1}{\bar{U}_T - \bar{U}_t} + 
\varepsilon I^{1}(t, z) + \varepsilon ^2I^{2}(t, z) + \cdots , 
\end{eqnarray}
where 
\begin{eqnarray*}
I^1(t, z) &=& 
-\frac{1}{(\bar{U}_T - \bar{U}_t)^2}\int ^T_tm(s, t, z)\bar{u}_s\, ds, \\
I^2(t, z) &=& 
-\frac{1}{(\bar{U}_T - \bar{U}_t)^2}\int ^T_t\left\{ \frac{1}{2}A_1(s, t, z) + A_2(s, t, z) - 2A_3(s, t, z)\right\} \bar{u}_s\, ds, \\
m(s, t, z) &=& \E [Z^{t, z}_s], \\
A_1(s, t, z) &=& \E [(Z^{t, z}_s)^2], \\
A_2(s, t, z) &=& \frac{1}{(\bar{U}_T - \bar{U}_s)^2}\E \left [\left (\int ^T_sm(r, s, Z^{t, z}_s)\bar{u}_r\, dr\right )^2\right ], \\
A_3(s, t, z) &=& \frac{1}{\bar{U}_T - \bar{U}_s}\int ^T_s\E [Z^{t, z}_sm(r, s, Z^{t, z}_s)]\bar{u}_r\, dr. 
\end{eqnarray*}
Substituting (\ref{asy_exp_W}) for (\ref{ODE_optimal}), we get the following second-order approximation formula for the optimal adaptive strategy: 
\begin{eqnarray}\nonumber 
x^\infty _t &=& X^\infty _t v_tW^\varepsilon (t, Z^{0, 0}_t)\\\label{approx_adaptive}
&\approx & 
X^\infty _t v_t\left\{ \frac{1}{\bar{U}_T - \bar{U}_t} + 
\varepsilon I^{1}(t, Z^{0, 0}_t) + \varepsilon ^2I^{2}(t, Z^{0, 0}_t) \right\} , 
\end{eqnarray}
or, equivalently, 
\begin{eqnarray*}
x^\infty _t &\approx & X_0\exp \left( -\int ^t_0v_s\left\{ \frac{1}{\bar{U}_T - \bar{U}_s} + 
\varepsilon I^{1}(s, Z^{0, 0}_s) + \varepsilon ^2I^{2}(s, Z^{0, 0}_s)\right\} ds\right) \\
&&\times 
v_t\left\{ \frac{1}{\bar{U}_T - \bar{U}_t} + 
\varepsilon I^{1}(t, Z^{0, 0}_t) + \varepsilon ^2I^{2}(t, Z^{0, 0}_t)\right\} . 
\end{eqnarray*}
To align the notation with $\hat{\bf x}^\mathrm {ant}$ and $\hat{\bf x}^\mathrm {stat}$, we also denote $\hat{\bf x}^\mathrm {adap} = (\hat{x}^\mathrm {adap}_t)_{0\leq t\leq T} = (x^\infty _t)_{0\leq t\leq T}$. We also get the approximation formula for $\hat {J}^\mathrm {adap}(X_0)$: 
\begin{eqnarray*}
\hat {J}^\mathrm {adap}(X_0) \approx 
X^2_0\left\{ \frac{\kappa }{2} + \tilde{\kappa }\left( 
\frac{1}{\bar{U}_T} + \varepsilon I^1(0, 0) + \varepsilon ^2I^2(0, 0) \right) \right\} . 
\end{eqnarray*}

Note again that the above derivation is only formal, so the accuracy of the approximation is not guaranteed at this stage. Therefore, we examine the accuracies and properties of the approximated adaptive optimal strategies by means of numerical experiments. We consider the case in which the noise process follows the Ornstein--Uhlenbeck (OU) process, that is, $\alpha (t, z) = -\rho z$ and $\beta (t, z) \equiv \sigma $ for some constants $\rho , \sigma > 0$. In this case, the approximation terms are given as 
\begin{eqnarray*}
I^1(t, z) &=& 
-\frac{z}{(\bar{U}_T - \bar{U}_t)^2}\int ^T_te^{-\rho (s-t)}\bar{u}_s\, ds, \\
I^2(t, z) &=& 
-\frac{1}{(\bar{U}_T - \bar{U}_t)^2}\int ^T_t\left\{ z^2e^{-2\rho (s-t)} + \frac{\sigma ^2}{2\rho }(1 - e^{-2\rho (s-t)})\right\} \left( \hat{U}^2_s - \frac{1}{2}\right) \bar{u}_s\, ds, 
\end{eqnarray*}
where 
\begin{eqnarray*}
\hat{U}_s = 1 - \frac{1}{\bar{U}_T - \bar{U}_s}\int ^T_se^{-\rho (r-s)}\bar{u}_r\, dr.
\end{eqnarray*}

We set the parameters as $\kappa = 0.0001, \tilde {\kappa } = 0.01$, $T = 1$, $X_0 = 10$, and $\sigma = 0.3$. For $\rho $, we examine the three patterns with $\rho = 0.3, 2$, and $5$. The parameter $\varepsilon $ is chosen in the interval $[0, 1]$. We also assume that $\bar{u}_t = 100$ throughout. We use numerical calculations involving the Euler--Maruyama approximation to compare the performances of three execution strategies, namely the expected VWAP strategy (\ref{def_static_VWAP}), the adaptive (approximated) optimal strategy (\ref{approx_adaptive}), and the exact VWAP strategy (\ref{def_exact_VWAP}).

Firstly, we examine the case of $\rho = 0.3$. When $\rho $ is small, the process $(v_t)_t$ fluctuates in a similar manner to geometric Brownian motion, thus the value of $\hat{J}^\mathrm {adap}(X_0)$ is expected to be close to $\hat{J}^\mathrm {stat}(X_0)$ (see Theorem~\ref{th_eg}). The results are summarized in 
Fig.~1, where we see that $\hat{J}^\mathrm {adap}(X_0)$ is quite similar to $\hat{J}^\mathrm {stat}(X_0)$. This result suggests that our approximation method is accurate even for $\varepsilon$ close to $1$. Interestingly, although the expected IS cost of $\hat{\bf x}^\mathrm {adap}$ is close to that of $\hat{\bf x}^\mathrm {stat}$, the forms of these strategies are different. 
Figure~2 shows the sample paths of the three strategies $\hat{\bf x}^\mathrm {stat}$, $\hat{\bf x}^\mathrm {adap}$, and $\hat{\bf x}^\mathrm {ant}$ for $\varepsilon = 0.3$. The fluctuation of $\hat{\bf x}^\mathrm {adap}$ is more similar to that of $\hat{\bf x}^\mathrm {ant}$ than that of $\hat {\bf x}^\mathrm {stat}$, especially when $t$ is small. As long as we follow adaptive strategies, we cannot observe the final value of $V_T$. Hence, the fluctuation of $\hat{\bf x}^\mathrm {adap}$ becomes unstable as $t$ approaches $T = 1$ to comply with the sell-off condition $\int ^T_0\hat{x}^\mathrm {adap}_t\, dt = X_0$. As a consequence, $\hat{\bf x}^\mathrm {adap}$ cannot improve the expected IS cost relative to $\hat{J}^\mathrm {stat}(X_0)$ in this case. 

\begin{figure}[!h]
\begin{center}
\includegraphics[height = 6cm,width=12cm]{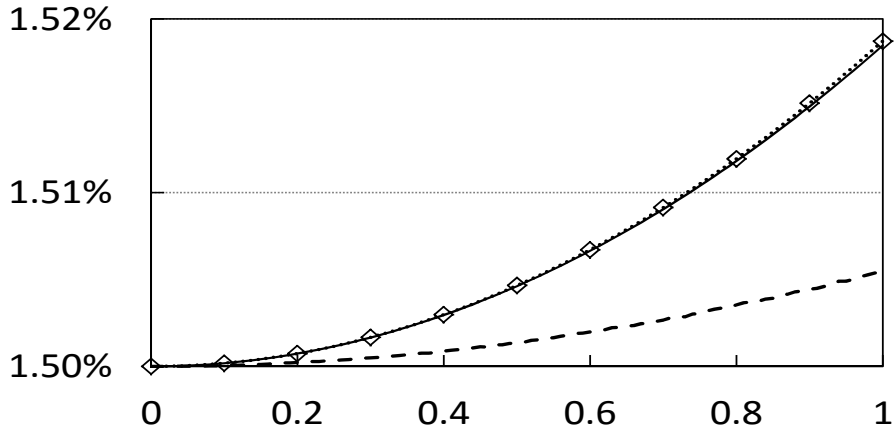}
\caption{Expected IS costs corresponding to the static/adaptive/anticipating optimal strategies via the parameter $\varepsilon \in [0, 1]$ for $\rho = 0.3$. The vertical axis corresponds to the cost value (solid line: $\hat{J}^\mathrm {adap}(X_0)$; diamond marked dotted line: $\hat{J}^\mathrm {stat}(X_0)$; dashed line: $\hat{J}^\mathrm {ant}(X_0)$). The horizontal axis corresponds to $\varepsilon $. }
\label{fig_rho=0.3}
\end{center}
\end{figure}

\begin{figure}[!h]
\begin{center}
\includegraphics[height = 6cm,width=12cm]{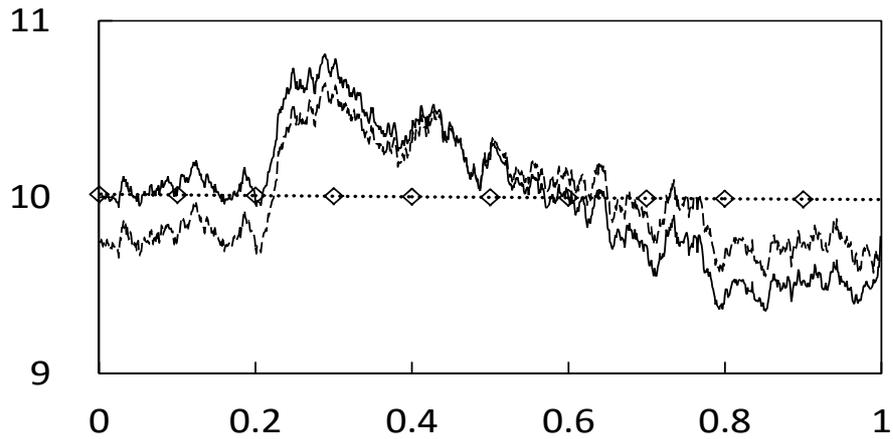}
\caption{Sample paths of the static/adaptive/anticipating optimal strategies for $\rho = 0.3$ and $\varepsilon = 0.3$. The vertical axis corresponds to the execution speed of each strategy (solid line: $\hat{\bf x}^\mathrm {adap}$; diamond marked dotted line: $\hat{\bf x}^\mathrm {stat}$; dashed line: $\hat{\bf x}^\mathrm {ant}$). The horizontal axis corresponds to time $t$. }
\label{fig_rho=0.3_sim}
\end{center}
\end{figure}

Next, we study the cases of $\rho = 2$ and $5$. 
Figure~3 shows comparisons of the values of $\hat{J}^\mathrm {stat}(X_0)$, $\hat{J}^\mathrm {adap}(X_0)$, and $\hat{J}^\mathrm {ant}(X_0)$ for $\rho = 2$ (left) and $5$ (right). In these cases, we find that $\hat{J}^\mathrm {adap}(X_0)$ is clearly smaller than $\hat{J}^\mathrm {stat}(X_0)$ when $\varepsilon $ is large. In particular, the difference between $\hat{J}^\mathrm {adap}(X_0)$ and $\hat{J}^\mathrm {stat}(X_0)$ becomes clearer with increase of the mean-reverting-speed parameter $\rho $. This is because, when $v_t$ fluctuates considerably, $v_t$ is expected to rapidly approach the mean reverting level; the adaptive strategy can take such information into account. Thus, in these cases, adaptive optimization works better than static optimization. Sample paths corresponding to the optimal strategies are shown in 
Fig.~4. As with the case of $\rho = 0.3$, we observe that $\hat {\bf x}^\mathrm {adap}$ fluctuates in tandem with $\hat {\bf x}^\mathrm {ant}$ until time $t$ approaches the sell-off time $T$.

\begin{figure}[!h]
\begin{center}
\includegraphics[height = 4cm,width=8cm]{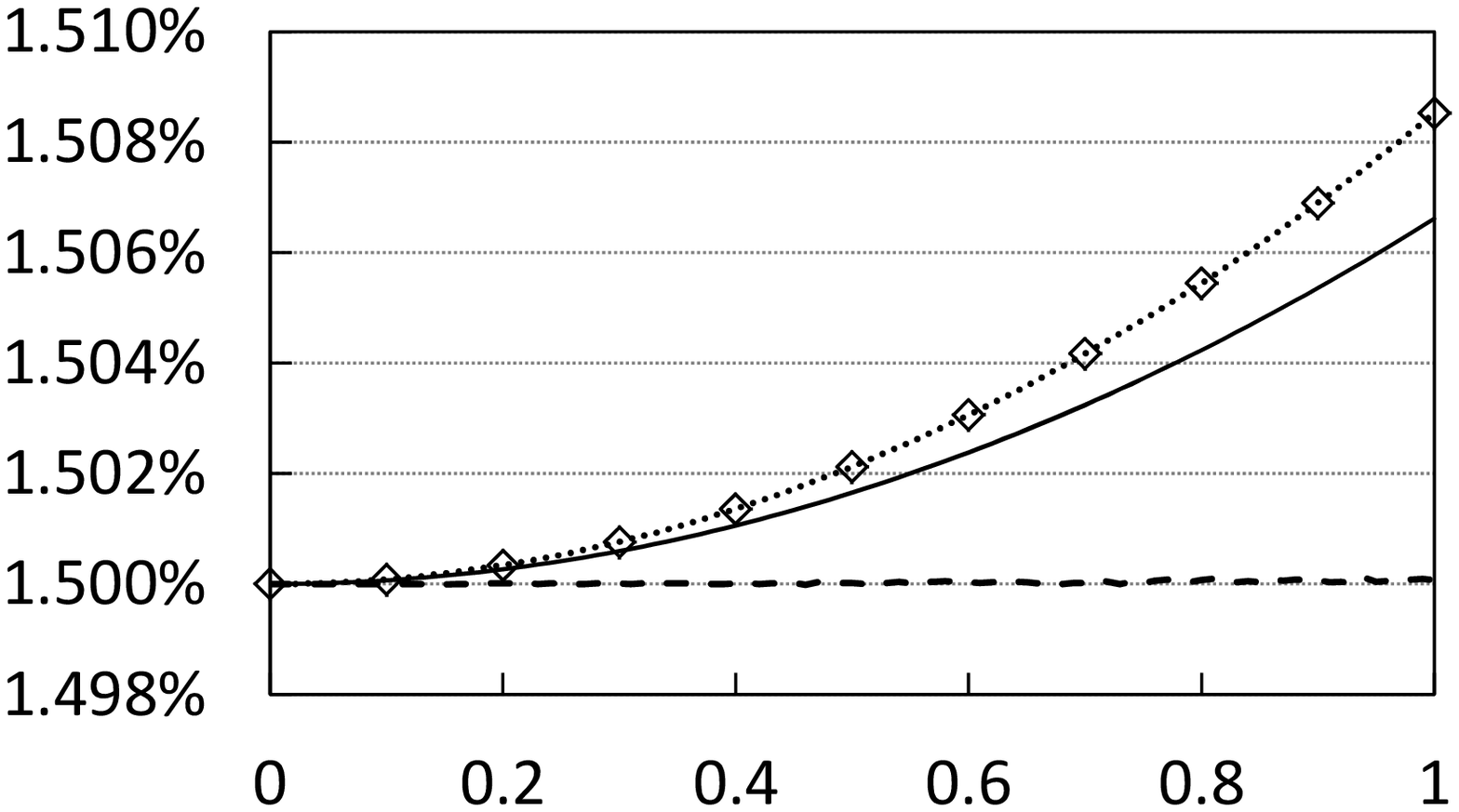}
\includegraphics[height = 4cm,width=8cm]{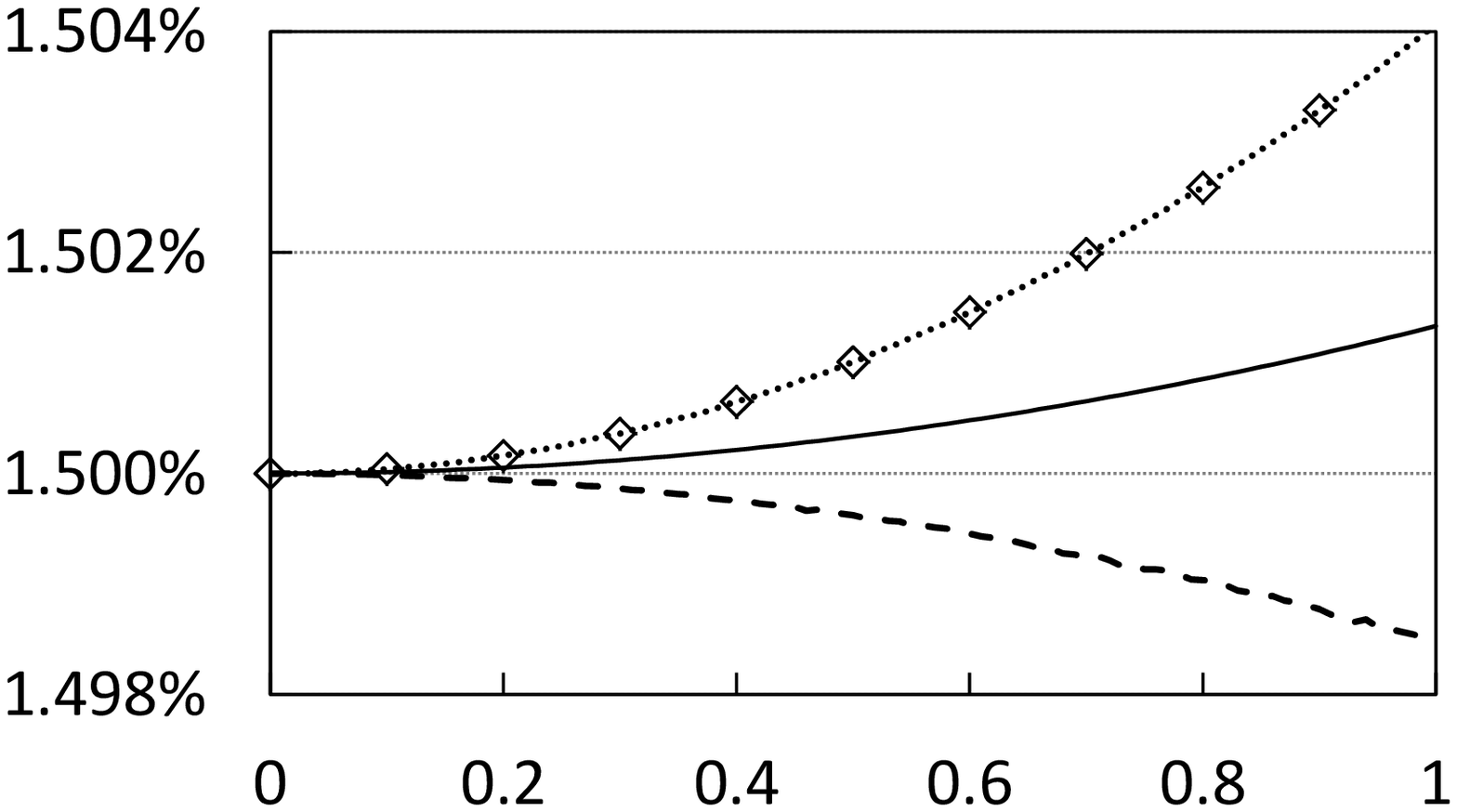}
\caption{Expected IS costs corresponding to the static/adaptive/anticipating optimal strategies via the parameter $\varepsilon \in [0, 1]$ (left: $\rho = 2$; right: $\rho = 5$). The vertical axes correspond to the cost value (solid line: $\hat{J}^\mathrm {adap}(X_0)$; diamond marked dotted line: $\hat{J}^\mathrm {stat}(X_0)$; dashed line: $\hat{J}^\mathrm {ant}(X_0)$). The horizontal axes correspond to $\varepsilon $. }
\label{fig_rho=2_5}
\end{center}
\end{figure}

\begin{figure}[!h]
\begin{center}
\includegraphics[height = 4cm,width=8cm]{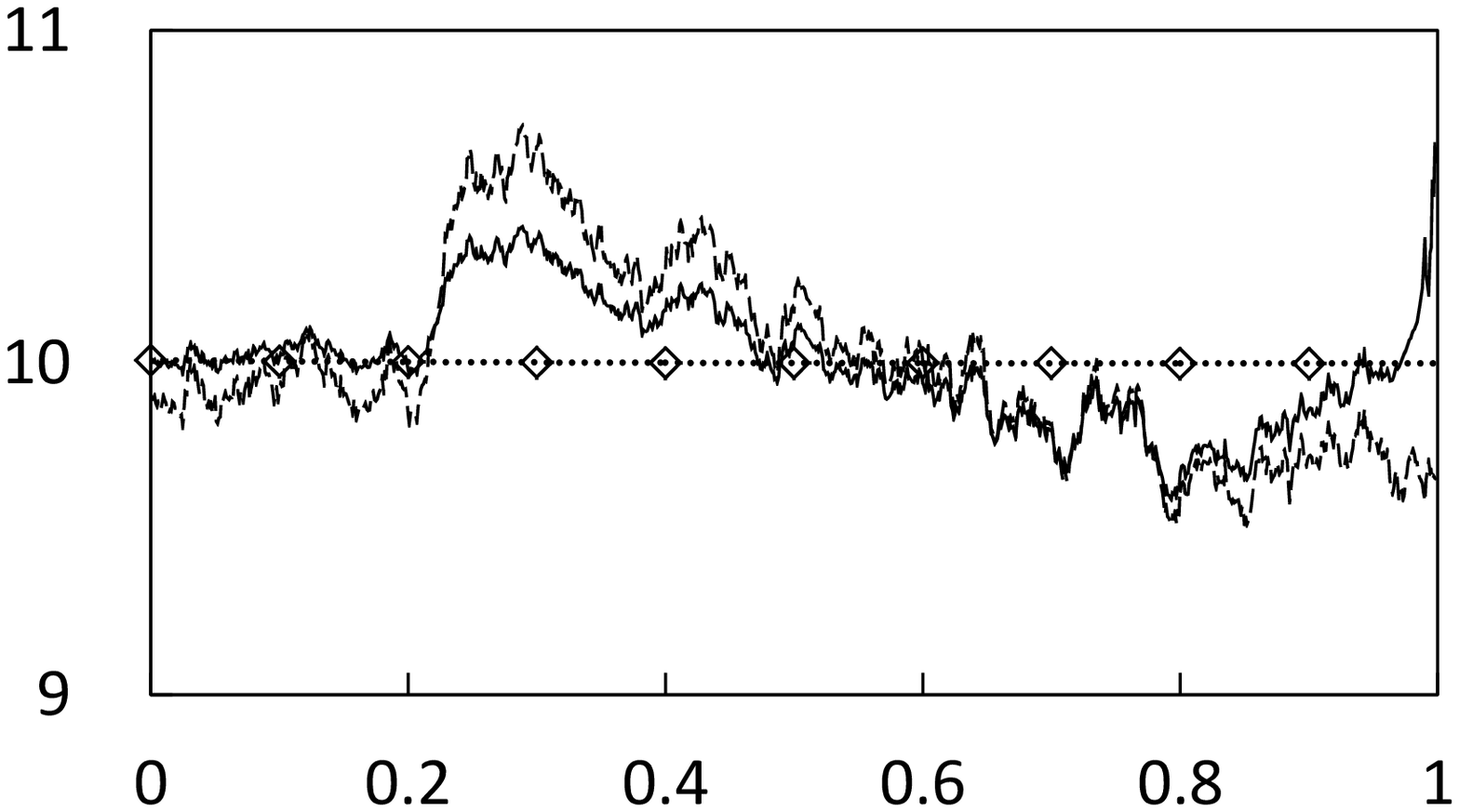}
\includegraphics[height = 4cm,width=8cm]{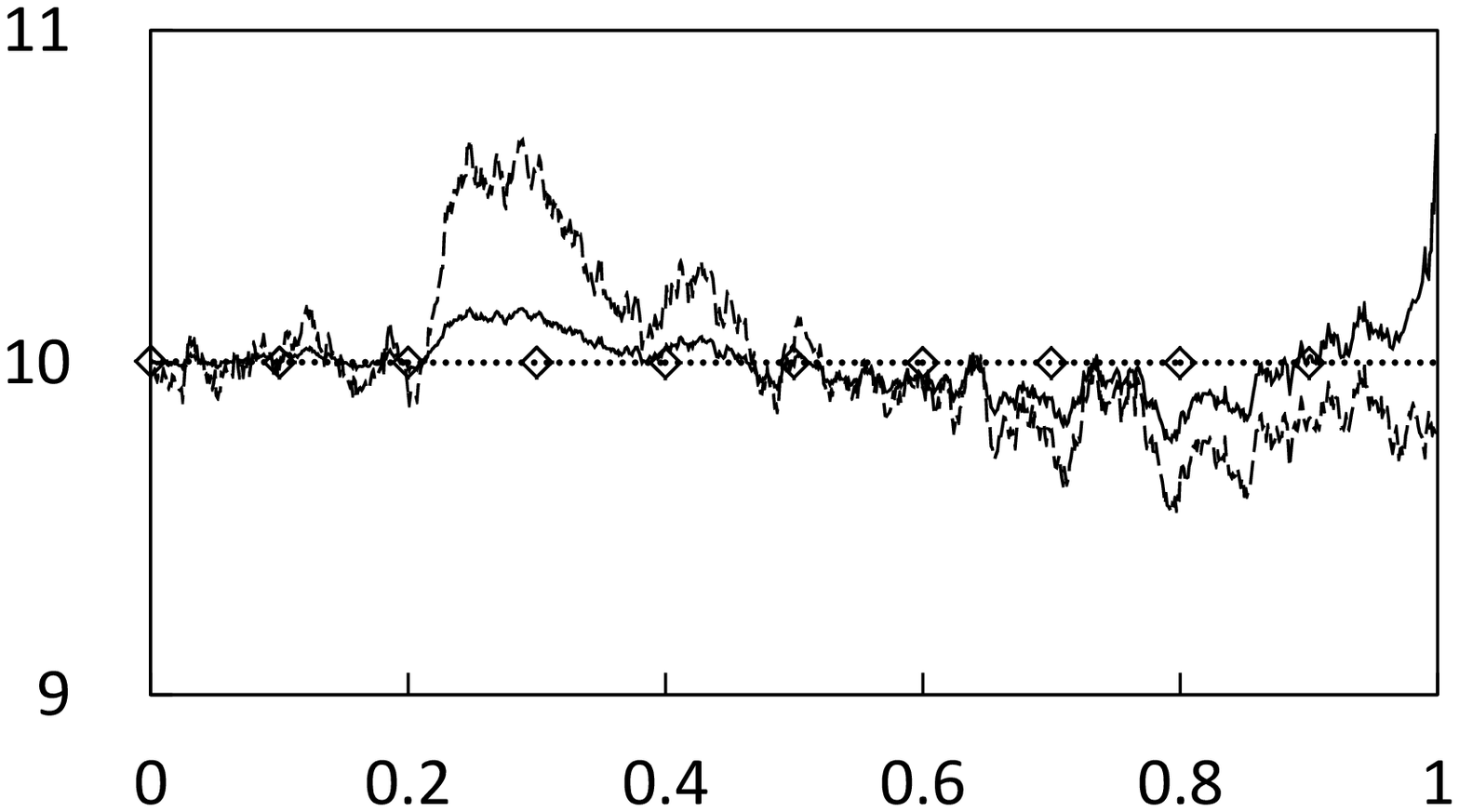}
\caption{Sample paths of the static/adaptive/anticipating optimal strategies for $\varepsilon = 0.3$ (left: $\rho = 2$; right: $\rho = 5$). The vertical axes correspond to the execution speed of each strategy (solid line: $\hat{\bf x}^\mathrm {adap}$; diamond marked dotted line: $\hat{\bf x}^\mathrm {stat}$; dashed line: $\hat{\bf x}^\mathrm {ant}$). The horizontal axes correspond to time $t$. }
\label{fig_rho=2_5_sim}
\end{center}
\end{figure}

\section{Concluding remarks}\label{sec_conclusion}
In this study, we have treated the optimal execution problem in the generalized AC model such that the temporary MI function depends on the market trading volume. We used the verification theorem to derive an adaptive optimal execution strategy, and as an application we showed that the expected VWAP strategy is optimal when the trading-volume process is given as the time-dependent Black--Scholes model. 

It is often found in studies on optimal execution problems for a risk-neutral trader (e.g., \cite{Alfonsi-Fruth-Schied, Gatheral-Schied_AC1, Kato_FS, Kuno-Onishi}, and \cite{Schied-Zhang}) that the adaptive optimal strategy is given by a deterministic process. Hence, there is little incentive to construct a dynamic strategy by updating the execution speed using current information about the random fluctuations of market data with time. This phenomenon is also true in our case in the time-dependent Black--Scholes framework. 

However, our numerical experiments implied that the adaptive optimal strategy is not static in general. When the trading-volume process was given as the geometric OU process, the dynamic (adaptive) optimization improved the expected IS cost compared with the case of static optimization. In particular, when the mean reverting speed was high, we observed a clearer difference between $\hat{J}^\mathrm {adap}(X_0)$ and $\hat{J}^\mathrm {stat}(X_0)$. 

As mentioned in Remark~\ref{rem_volume_time}, our model can be interpreted as the AC model defined on the volume timeline. To see this, the linearity of $g$ is essential. Also, as mentioned in Remarks~\ref{rem_tilde_g_general} and \ref{rem_nonlinear_tilde_g}, our results also work for generally shaped $\tilde{g}$ up to the form (\ref{tilde_g_general}), whereas we cannot avoid the linearity of $g$ in any case. However, 
Kato (2014a,b, 2015, 2017) succeeded in tackling an optimal execution problem in which the MI functions were nonlinear. The challenge remains to study the AC model with a nonlinear form of $g$ that depends on the market trading volume.

\appendix 

\section{Proof of Theorem~\ref{th_verification}}\label{sec_proofs}

In this section, we always assume [A1]--[A3]. The following lemma is immediately obtained by [A1] and Corollary~2.5.10 in \cite{Krylov} for $v^{\pm m}_t$. 

\begin{lemma}\label{lem_moment}
For each $m\geq 1$, there is a constant $C_m > 0$ such that 
\begin{eqnarray*}
&&\E \left[ \sup _{0\leq t\leq T}v^{m}_t\right] \leq C_m(1 + v^{m}_0), \\
&&\E \left[ \sup _{0\leq t\leq T}v^{-m}_t\right] \leq C_m(1 + v^{-m}_0). 
\end{eqnarray*}
\end{lemma}

Define 
\begin{eqnarray*}
J^\lambda (t, X, v) =  \inf _{{\bf x}\in \tilde {\mathcal {A}}(t, X)}\E \left[ 
\int ^T_t\frac{x^2_s}{v_s}ds + \frac{\lambda }{v_t}\left( X - \int ^T_tx_s\, ds \right)^2\right], 
\end{eqnarray*}
where $(v_s)_{t\leq s\leq T}$ is a solution to 
\begin{eqnarray*}
dv_s = \hat{b}(s, v_s)\, ds + \hat{\sigma }(s, v_s)dB_s \ (s\geq t), \ \ v_t = v 
\end{eqnarray*}
and 
\begin{eqnarray*}
\tilde{\mathcal {A}}(t, X) = 
\left\{ {\bf x} = (x_s)_{t\leq s\leq T}\ ;\ (\mathcal {F}_s)_{t\leq s\leq T}\mbox {-adapted}, 
x_s\geq 0 \ \mbox { and } \int ^T_tx_s\, ds \leq  X\ \mbox {a.s.}  \right\} . 
\end{eqnarray*}
We also put $\bar{J}^\lambda (t, X, v) = X^2W^\lambda (t, v)$. 
%
%
%
%

\begin{prop}\label{prop1}$J^\lambda (t, X, v)\geq \bar{J}^\lambda (t, X, v)$. 
\end{prop}

\begin{proof}
Fix any ${\bf x}\in \tilde {\mathcal {A}}(t, X)$ and set 
\begin{eqnarray}\label {def_X_t}
X_s = X - \int ^s_tx_r\, dr, \ \ t\leq s\leq T. 
\end{eqnarray}
For each $R > 0$, set $\tau _R = \inf \{ s\geq t\ ; \ |\log v_s|\geq R \}$ ($\inf \phi \equiv \infty $). Applying Ito's formula, we see that 
\begin{eqnarray}\nonumber 
&&\int ^{T\wedge \tau _R}_t\frac{x^2_s}{v_s}ds + \bar{J}^\lambda (T\wedge \tau _R, X_{T\wedge \tau _R}, v_{T\wedge \tau _R})\\\nonumber 
&=& 
\bar{J}^\lambda (t, X, v) + 
\int ^{T\wedge \tau _R}_t\Bigg [ \Bigg \{ \frac{\partial }{\partial s} + \hat{b}\frac{\partial }{\partial v} + \frac{1}{2}\hat{\sigma }^2\frac{\partial ^2}{\partial v^2}\Bigg \} \bar{J}^\lambda (s, X_s, v_s) + \frac{x^2_s}{v_s} - x_s\frac{\partial }{\partial X}\bar{J}^\lambda (s, X_s, v_s)\Bigg ]ds\\\label{temp_calc1}
&& + (\mbox {martingale}),
\end{eqnarray}
including the case of ``$\infty = \infty $'' (note that both sides of the above equality may diverge to $\infty $ because of the integral of $x^2_s/v_s\geq 0$). 
We notice that 
\begin{eqnarray}\nonumber 
\frac{x^2}{v} - x\frac{\partial }{\partial X}\bar{J}^\lambda (s, X, v) 
&=& 
\frac{1}{v}\left( x - vXW^\lambda (s, v) \right) ^2 - vX^2(W^\lambda (s, v))^2
\\\label{temp_min}
&\geq & 
-vX^2(W^\lambda (s, v))^2 
\end{eqnarray}
for each $x\geq 0$ (note that (\ref{temp_min}) is also valid for all $x\in \Bbb {R}$). Combining this with [A2] and (\ref{temp_calc1}), we obtain 
\begin{eqnarray}\label{temp_calc2}
\E \left[ \int ^{T\wedge \tau _R}_t\frac{x^2_s}{v_s}ds\right]  + \E \left[ \bar{J}^\lambda (T\wedge \tau _R, X_{T\wedge \tau _R}, v_{T\wedge \tau _R})\right] 
\geq \bar{J}^\lambda (t, X, v). 
\end{eqnarray}
It holds from Lemma~\ref{lem_moment} that $\lim _{R\rightarrow \infty }\tau _R\geq T$ a.s. Therefore, the first term on the left-hand side of (\ref{temp_calc2}) converges to $\E \left [ \int ^T_t(x^2_s/v_s)\, ds \right ]$ as $R\rightarrow \infty $ by the monotone convergence theorem. As for the second term, we observe that 
\begin{eqnarray*}
\E \left[ \bar{J}^\lambda (T\wedge \tau _R, X_{T\wedge \tau _R}, v_{T\wedge \tau _R})\right] = 
\lambda \E \left[ \frac{X^2_T}{v_T}\ ; \ \tau _R\geq T\right]  + 
\E \left[ \bar{J}^\lambda (\tau _R, X_{\tau _R}, v_{\tau _R})\ ; \ \tau _R < T\right] , 
\end{eqnarray*}
that 
\begin{eqnarray*}
\E \left[ \frac{X^2_T}{v_T}\ ; \ \tau _R\geq T\right] \ \longrightarrow \ 
\E \left[ \frac{X^2_T}{v_T}\right], \ \ R\rightarrow \infty 
\end{eqnarray*}
and that 
\begin{eqnarray*}
0&\leq & 
\E \left[ \bar{J}^\lambda (\tau _R, X_{\tau _R}, v_{\tau _R})\ ; \ \tau _R < T\right] \leq 
X^2C_\lambda \E \left [\left (1 + \sup _{t\leq s\leq T}v^{m_\lambda }_s + \sup _{t\leq s\leq T}v^{-m_\lambda }_s\right )1_{\{ \tau _R < T\} }\right ]\\
&\longrightarrow & 
0, \ \ R\rightarrow T 
\end{eqnarray*}
by using Lemma~\ref{lem_moment} and the dominated convergence theorem. Combining these with (\ref{temp_calc2}), we get that 
\begin{eqnarray*}
\E \left[ 
\int ^T_t\frac{x^2_s}{v_s}ds + \frac{\lambda X^2_T}{v_t}\right] \geq \bar{J}^\lambda (t, X, v). 
\end{eqnarray*}
Since ${\bf x}$ is arbitrary, we obtain the assertion. 
\end{proof}

\begin{prop}\label{prop2}$J^\lambda (t, X, v)\leq \bar{J}^\lambda (t, X, v)$. 
\end{prop}

\begin{proof}
Set 
\begin{eqnarray*}
x^\lambda _s = X\exp \left( -\int ^s_tv_rW^\lambda (r, v_r)\, dr \right) v_sW^\lambda (s, v_s). 
\end{eqnarray*}
Then ${\bf x}^\lambda = (x^\lambda _s)_{t\leq s\leq T}$ is $(\mathcal {F}_s)_{t\leq s\leq T}$-adapted, non-negative, and 
\begin{eqnarray*}
X^\lambda _s := X - \int ^s_tx^\lambda _r\, dr = 
X\exp \left( -\int ^s_tv_rW^\lambda (r, v_r)\, dr \right) \leq X, \ \ t\leq s\leq T, 
\end{eqnarray*}
hence ${\bf x}^\lambda \in \tilde{\mathcal {A}}(t, X)$ holds. 
We observe 
\begin{eqnarray*}
\frac{v_s}{2}\frac{\partial }{\partial X}\bar{J}^\lambda (s, X^\lambda _s, v_s) = 
X^\lambda _sv_sW^\lambda (s, v_s) = x^\lambda _s
\end{eqnarray*}
to arrive at 
\begin{eqnarray}\label{temp_opt}
\frac{(x^\lambda _s)^2}{v_s} - x^\lambda _s\frac{\partial }{\partial X}\bar{J}^\lambda (s, X^\lambda _s, v_s) = 
-v_s(X^\lambda _sW^\lambda (s, v_s))^2. 
\end{eqnarray}
By the same calculation as in the proof of Proposition~\ref{prop1}, replacing (\ref{temp_min}) with (\ref{temp_opt}), we see that 
\begin{eqnarray*}
J^\lambda (t, X, v)\leq 
\E \left[ 
\int ^T_t\frac{(x^\lambda _s)^2}{v_s}ds + \frac{\lambda }{v_t}\left( X^\lambda _s \right)^2\right] = 
\bar{J}^\lambda (t, X, v). \qedhere 
\end{eqnarray*}
\end{proof}

From Propositions~\ref{prop1} and \ref{prop2}, we obtain (\ref{def_J_lambda}) and we see that ${\bf x}^\lambda \in \tilde{\mathcal {A}}(0, X_0)$ defined by (\ref{def_penalized_optimizer}) is an optimizer of $J^\lambda (X_0)$. 

By the definition, we see that for each $(t, v)\in [0, T)\times (0, \infty )$, 
the family of functions $W^\lambda (t, v) = J^\lambda (t, 1, v)$, $\lambda > 0$, is non-negative and monotone increasing with respect to $\lambda $. Furthermore, it holds that $\sup _{\lambda }W^\lambda (t, v) < \infty $. Indeed, setting ${\bf x} = (x_s)_{t\leq s\leq T}\in \mathcal {A}(t, 1)$ as $x_s = 1/(T - t)$, we see that 
\begin{eqnarray*}
J^\lambda (t, 1, v) \leq  \E \left[ 
\int ^T_t\frac{x^2_s}{v_s}ds + \frac{\lambda }{v_t}\left( 1 - \int ^T_tx_s\, ds \right)^2\right] \leq 
\frac{1}{T - t}\E \left[ \sup _{0\leq s\leq T}v^{-1}_s\right] \leq \frac{C_1}{T-t}
\end{eqnarray*}
because of Lemma~\ref{lem_moment}. Therefore, the limit $W^\infty (t, v) = \lim _{\lambda \rightarrow \infty }W^\lambda (t, v)$ exists for each $(t, v)\in [0, T)\times (0, \infty )$. Hence, the monotone convergence theorem implies that the limit 
\begin{eqnarray}\label{conv_x}
x^\infty _t = \lim _{\lambda \rightarrow \infty }x^\lambda _t = 
X_0\exp \left( -\int ^t_0v_rW^\infty (s, v_s)\, ds \right) v_tW^\infty (t, v_t) 
\end{eqnarray}
also exists for almost all $(t, \omega )\in [0, T]\times \Omega $. 

To see the admissibility of ${\bf x}^\infty = (x^\infty _t)_{0\leq t\leq T}$, we observe that 
\begin{eqnarray*}
\lambda \E \left[ \frac{1}{v_T}\left( X_0 - \int ^T_0x^\lambda _t\, dt \right)^2\right] \leq 
J^\lambda (X_0)\leq \frac{C_1X^2_0}{T}
\end{eqnarray*}
to arrive at 
\begin{eqnarray*}
\lim _{\lambda \rightarrow \infty }\E \left[ \frac{1}{v_T}\left( X_0 - \int ^T_0x^\lambda _t\, dt \right)^2\right] = 0. 
\end{eqnarray*}
Therefore, Fatou's lemma implies that 
\begin{eqnarray}\label{temp_Fatou}
\liminf _{\lambda \rightarrow \infty }\left( X_0 - \int ^T_0x^\lambda _t\, dt \right)^2 = 0 \ \ \mathrm {a.s.}
\end{eqnarray}
Moreover, (\ref{conv_x}) and the dominated convergence theorem imply that $\lim _{\lambda \rightarrow \infty }\int ^T_0x^\lambda _t\, dt = \int ^T_0x^\infty _t\, dt$ a.s. Combining this with (\ref{temp_Fatou}), we obtain $\int ^T_0x^\infty _t\, dt = X$ a.s., which implies that ${\bf x}^\infty \in \mathcal {A}(X_0)$. 

Finally, we show the optimality of ${\bf x}^\infty $. We have that 
\begin{eqnarray*}
J^\lambda (X_0) = X^2_0W^\lambda (0, v_0) \geq \E \left[ \int ^T_0\frac{(x^\lambda _t)^2}{v_t}dt \right] . 
\end{eqnarray*}
Letting $\lambda \rightarrow \infty $, we obtain 
\begin{eqnarray}\label {temp_ineq}
X^2_0W^\infty (0, v_0) \geq \E \left[ \int ^T_0\frac{(x^\infty _t)^2}{v_t}dt \right] \geq J(X_0) 
\end{eqnarray}
by using Fatou's lemma again. However, (\ref{def_J_lambda}) and (\ref{rel_J}) imply that $X^2_0W^\infty (0, v_0)\leq J(X_0)$. Therefore, (\ref{temp_ineq}) holds, replacing ``$\geq $'' with ``$=$.'' Thus, the proof is complete. \qed 

\section{Volume-dependent permanent MI functions}\label{sec_permanent}

Until now, we have assumed that the permanent MI function $g$ is a linear function with respect to the execution speed and is independent of the market trading volume. Here, we study the case in which both $g$ and $\tilde{g}$ depend on the market trading volume, that is, 
\begin{eqnarray*}
g(v, x) = \frac{\kappa x}{v}, \ \ \tilde{g}(v, x) = \frac{\tilde{\kappa }x}{v}. 
\end{eqnarray*}
In this case, it holds that 
\begin{eqnarray}\label{value_fnc_permanent}
\hat{J}^\mathrm {adap}(X_0) = 
\inf _{{\bf x}\in \hat{\mathcal {A}}(0, X_0)}\E \left[ 
\int ^T_0\frac{\kappa X_tx_t + \tilde{\kappa }x^2_t}{v_t} dt \right] , 
\end{eqnarray}
where $\hat{\mathcal {A}}(t, X)$ is defined as in (\ref{def_adm_str2}) 
to treat the opportunity when admissible strategies include purchasing orders (see Remark~\ref{rem_non_negativity}). Here, we confirm the boundedness of the process of the security shares held for each strategy in $\hat {\mathcal {A}}(t, X)$. 

\begin{lemma}\label{bdd_X}For each ${\bf x} = (x_s)_{t\leq s\leq T}\in \hat {\mathcal {A}}(t, X)$, define $(X_s)_{t\leq s\leq T}$ by {\rm (\ref{def_X_t})}. Then it holds that $\mathop {\rm esssup}_{s, \omega }|X_s(\omega )| < \infty $. 
\end{lemma}

\begin{proof} 
Denote $x^+_s = \max \{ x_s, 0 \}$ and $x^-_s = -\min \{ x_s, 0 \}$ so that $x_s = x^+_s - x^-_s$ and $|x_s| = x^+_s + x^-_s$. By the assumption, we have $0\leq \eta ^-_s\leq CT$ and $X_s = X - \eta ^+_s + \eta ^-_s \leq X - \eta ^+_s + CT$ for each $s\in [t, T]$ a.s., where 
\begin{eqnarray*}
\eta ^\pm _s = \int ^s_0x^\pm _r\, dr, \ \ 
C = \mathop {\rm esssup}_{s, \omega }x^-_s(\omega ) < \infty . 
\end{eqnarray*}
In particular, it follows that $0 \leq X_T \leq X - \eta ^+_T + CT$, hence $0\leq \eta ^+_s \leq \eta ^+_T \leq X + CT$ for each $s\in [t, T]$ a.s. Now our assertion is obvious. 
\end{proof}


When $(v_t)_{0\leq t\leq T}$ is given as geometric Brownian motion (\ref{GBM}), a similar argument to that in Appendix~\ref{sec_proofs}, replacing the assumption (\ref{cond_A2ii}) with (\ref{cond_A2ii_gen}), leads us to the following theorem. 

\begin{theorem}\label{th_permanent}
Put $\tilde{\mu } = \mu - \sigma ^2/2$ and $D = \tilde{\mu }^2 - 2\tilde{\mu }\kappa / \tilde{\kappa }$. An optimizer $\hat {\bf x}^\mathrm {adap} = (\hat{x}^\mathrm {adap}_t)_{0\leq t\leq T}$ of $(\ref{value_fnc_permanent})$ is given as follows. 
\begin{itemize}
 \item [\rm {(i)}] If $D < 0$ and $\gamma T < 2\pi $, 
\begin{eqnarray*}
\hat{x}^\mathrm {adap}_t = 
\frac{X_0e^{\tilde{\mu }t/2}}{2\sin (\gamma T/2)}\left\{ 
\gamma \cos \left( \frac{\gamma }{2}(T - t) \right) - 
\tilde{\mu }\sin \left( \frac{\gamma }{2}(T - t) \right) \right\} . 
\end{eqnarray*}
 \item [\rm {(ii)}] If $D = 0$, 
\begin{eqnarray*}
\hat{x}^\mathrm {adap}_t = 
X_0e^{\tilde{\mu }t/2}\left\{ \frac{1}{T} - \frac{\tilde{\mu }}{2}\left( 1 - \frac{t}{T}\right)  \right\} . 
\end{eqnarray*}
 \item [\rm {(iii)}] If $D > 0$, 
\begin{eqnarray*}
\hat{x}^\mathrm {adap}_t = 
\frac{X_0}{2(e^{\gamma T} - 1)}\left\{ 
(\tilde{\mu }+ \gamma )e^{(\tilde{\mu } + \gamma )t/2} - (\tilde{\mu } - \gamma )e^{(\tilde{\mu } - \gamma )t/2 + \gamma T} \right\} . 
\end{eqnarray*}
\end{itemize}
\end{theorem}

The above theorem implies that the optimal strategy is static in each case, but is no longer the VWAP strategy. Moreover, in some cases, the optimal execution speed may become negative (i.e., the optimal selling strategy contains buying behaviors).

\section*{Acknowledgment}
The author would like to thank the anonymous referee for many valuable comments and suggestions that have improved the quality of the paper.

\end{document}